\documentclass[12pt]{article}       %--- LATEX 2e base
\usepackage{latexsym}
\usepackage{wrapfig}
\usepackage{algorithm}
\usepackage{amssymb, amsmath}
\usepackage[noend]{algcompatible}
\usepackage{subcaption}
\usepackage{amsmath}
\usepackage{mathrsfs}
\usepackage{verbatim}
\usepackage{breqn}
\usepackage{color}
\usepackage{yhmath}
\usepackage{fullpage}
\usepackage{afterpage}
\usepackage{epsfig}
\usepackage{titlesec}

\renewcommand{\baselinestretch}{1.4}
\newtheorem{definition}{Definition}
\newtheorem{theorem}{Theorem}
\newtheorem{lemma}{Lemma}
\newtheorem{corollary}{Corollary}
\newtheorem{property}{Property}
\newtheorem{observation}{Observation}

\newenvironment{proof}           
{
  \noindent{\bf Proof.} 
}
{
  $ $\newline
}

\definecolor{mygray}{gray}{0.3}

\newtheorem{algX}{Algorithm}
                                 {\end{sf}\par\noindent
                                 --- End of Procedure ---
                                 \end{em}\end{algX}}

\newcommand{\ee}{\mathcal{E}}

\renewcommand{\v}[1]{\mathcal{V}_{#1}}

\renewcommand{\l}[3]{
  \ifstrequal{#1}{}
    {\ell^{#3}_{#2}}
    {\ell^{#3}_{#2}({#1})}
}

\newcommand{\departurecurve}{\mathcal{D}}

\newcommand{\steinerpoints}{\mathcal{S}}

\newcommand{\alignset}{\mathcal{L}}

\newcommand{\weight}{\mathcal{W}}
\newcommand{\weights}[2]{\mathcal{W}({#1}, #2)}

\newcommand{\invalidinterval}{\mathcal{B}}

\newcommand{\robotpath}{\mathcal{R}}

\newcommand{\T}{\lambda}

\newcommand{\aus}[2]{\overline{A}(t)}

\newcommand{\fus}[2]{\overline{Aprx}(t)}

\newcommand{\vr}{\mathcal{V}_{max}}

\newcommand{\dd}{\beta}

\begin{document}

\title{Query Shortest Paths Amidst Growing Discs}

\author{Arash Nouri and J\"org-R\"udiger Sack 
\\  arash,sack@scs.carleton.ca \\ \\ Carleton University, Ottawa, Canada}

\maketitle

\begin{abstract}

The determination of collision-free shortest paths among growing discs has previously been  studied for discs with fixed growing rates  \cite{shortest_path_in_growing_disc_referred, growing_discs_overmars, yi-thesis}. 
Here, we study a more general case of this problem, where: (1) the speeds at which the discs are growing are polynomial functions of degree $\dd$, and (2) the source and destination points are given as query points. We show 
how to preprocess the $n$ growing discs so that, for two given query points $s$ and $d$, a shortest path from $s$ to $d$ can be found in $O(n^2 \log (\dd n))$ time.
The preprocessing time of our algorithm is $O(n^2 \log n + k \log k)$ where $k$ is the number of intersections between the growing discs and the  tangent paths (straight line paths which touch the boundaries of two growing discs). 
We also prove that $k \in O(n^3\dd)$.
 \end{abstract}

%!TEX root = ./qtmp-main.tex

\section{Introduction} 
 
Motion planning is a fundamental problem that arises in a number of  application areas including robotics, navigation and GIS. Often environments are not static and therefor motion planning algorithms which account for the dynamic nature of the obstacles attracted interest in computational geometry (see e.g. \cite{Mitchell1, Shortest-path-and-networks, planning-algorithms}). 
 
Path planning in dynamic environments was first studied by adding time as a dimension to the configuration space \cite{reif,erdmann}.
Erdmann and Lozano-Pérez \cite{erdmann} proposed a time discretization method to create a sequence of configuration space slices. Each slice represents the configuration space at a specific time instance.
They then solved the path planning problem in each slice and joined the adjacent solutions
to obtain a global solution. in many cases, in order to plan a path,
future trajectories of the moving obstacles are estimated using their previous velocities \cite{vandenberg,fiorini,high-speed-autonomous}. In these methods, when the obstacle trajectories change, the shortest path may need to be re-planned. This could be a slow and costly operation in highly dynamic
environments.
The challenge is  increased even further, if the obstacles' motions are unpredictable. Overmars and van den Berg
\cite{path-finding,shortest_path_in_growing_disc_referred} modeled such obstacles by discs, that grow over time.
Since the obstacles are located inside the growing discs at any time, the shortest paths are guaranteed
to be collision-free, no matter how often the obstacles change
their velocities in the future.
In this model,
each disc has a fixed growing
velocity which does not exceed the maximum
velocity of the robot. 
The goal is to find a collision-free time-minimal path from a source point to a destination point. This problem also finds applications in dynamic environments where the obstacles can be modeled as growing discs. Furthermore, in the context of video games, shortest path computations involving growing discs are carried out to model avoidance of obstacles whose maximum speed is known but not the direction of their movements.

In this paper, we study the query time-minimal path problem among growing discs (QTMP for short): Given a point robot $r$ moving with maximum velocity $\vr$ and a set of discs $\mathscr{D}=\{C_1, ..., C_n\}$, where the
radius of each disc $C_i$  at time $t$ grows with velocity $\v{i}(t)$ and $\v{i}(t)<\vr$. Let $\v{i}:T \rightarrow [0, \vr)$ be a polynomial function of degree $\dd$. In this setting, a query consists of a pair $(s, d)$, where $s$ is the source point and $d$ is the destination point. The problem is to
 determine a collision-free time-minimal path for a point robot, from $s$ to $d$. 
 
\textbf{Previous work}. 
Previous research focused on time-minimal path problems among growing discs, when $s$ and $d$ are fixed.
 Chang et al. \cite{chang} proposed an $O(n^2\log n)$ time algorithm for a special problem instance when $\v{i}(t)=0$ (i.e., the discs do not grow).
They proved that their algorithm is computable for arbitrary algebraic inputs.
% To our best knowledge, no faster algorithm exists for this case.
Overmars and van den Berg
\cite{path-finding,shortest_path_in_growing_disc_referred} gave an $O(n^3\log n)$ algorithm for the case where the discs grow with same constant velocity ($\v{i}(t) = c$, for some constant $c$). They used a
3-dimensional model where the $x$ and $y$ axes represent
location and the $z$ axis describes time. Then,
each disc  obstacle is modeled by a cone which
shows the obstacles' growth over time. In this approach, a $z$-monotone shortest path from the source  to the destination,
consists of straight line segments 
which are tangent to pairs of
cones and spiral segments on the surface of a cones. The projection of this path on the $xy$ plane is a time-minimal path between $s$ to $d$.
In \cite{yi-thesis}, Yi et al. improved this approach and presented an  $O(n^2\log n)$ time algorithm. They also gave an $O(n^3\log n)$ time algorithm for different but fixed velocity discs ($\v{i}(t) = c_i$, for some constant $c_i$ per disc $C_i$).  In this method, the tangents between the discs are identified as moving line segments. By calculating the times at which the endpoints of tangents meet, the collision-free tangents are identified. Then, a time-minimal path is built from the collision-free tangents. However, they did not take into account that these endpoints may meet more than once. 
A summary of the previous work is presented in Table \ref{pre-works}.

\begin{center}
\centering
  \begin{tabular}{ |c|c|c|c|c|  }

   \hline
   Growth rates  & Velocity functions   & Time complexity & Space & Reference\\
   \hline
   Static discs &$\v{i}(t) = \v{j}(t) = 0$ &  $O(n^2 \log n)$ & $O(n^2)$ & \cite{chang}\\
   Same speed discs &$\v{i}(t) = \v{j}(t) = c$ &  $O(n^3 \log n)$ &  $O(n^3)$ & \cite{path-finding,shortest_path_in_growing_disc_referred}\\
   Same speed discs &$\v{i}(t) = \v{j}(t) = c$ &  $O(n^2 \log n)$ & $O(n^2)$ & \cite{yi-thesis}\\
   Various speed discs  &$\v{i}(t) = c $ and $ \v{j}(t) = c'$ &  $O(n^3 \log n)$ & $O(n^3)$ & \cite{yi-thesis}\\
   \hline
  \end{tabular}
\captionof{table}{Note that $c$ and $c'$ are constant values where $1 \leq c, c' < \vr$.} \label{pre-works} 
\end{center}

% \begin{center}
%   \begin{tabular}{ |c|c|c|c|  }
 
%    \hline
%    Growing rates  & Space & Query Time & Preprocessing\\
%    \hline
%    $\v{i}(t)$ is a polynomial function of degree $\dd$  & $O(n^2 + k)$ &  $O(n^2 \log (kn))$ & $O(n^2 \log n + k \log k)$\\
%    \hline
%    $\v{i}(t) = c $ and $ \v{j}(t) = c'$  & $O(n^3)$ &  $O(n^2 \log n)$ & $O(n^3 \log n)$\\
%    \hline
%    $\v{i}(t) < \v{j}(t)$ where $i < j$ (Future Work) & $O(n^3)$ &  $O(n^2 \log n)$ & $O(n^3 \log n)$\\
%    \hline
%    $\v{i}(t) = \v{j}(t)$ where $i \not= j$ (Future Work) & $O(n^2)$ &  $O(n^2 \log n)$ & $O(n^2 \log n )$\\
%    \hline
%   \end{tabular}
% \captionof{table}{Note that $k \in O(n^3 \dd)$.} \label{pre-works} 
% \end{center}

\textbf{Contributions}. 
In this paper we present an algorithm for QTMP problem while has $O(n^2 \log (kn))$ query time after $O(n^2 \log n + k \log k)$ preprocessing time, where $k$ is the size of the arrangements formed by the \textit{departure curves}\footnote{ Refer to Section \ref{section-preliminary} for the formal definition.}. Not only does the algorithm generalizes the velocity functions to polynomial of degree $\dd$, it also solve the problem in query mode setting which decreases the time complexity of the existing algorithm \cite{yi-thesis} for the problem with different, but constant velocities (when $\dd = 0$) by a factor of $n$. 
We also establish an upper bound of  $O(n^3\dd)$ for $k$.

We obtain this result by associating a graph called the \textit{adjacency graph}$^1$
% \footnote{ Refer to Section \ref{section-preliminary} for the formal definition.} 
with the growing discs. Each edge in the adjacency graph represents a robot's path, which is collision-free at some specific time intervals. In the preprocessing phase, we construct some data structures, such that we can quickly determine if a query path is collision-free. Then, we run Dijkstra's algorithm to construct a time-minimal path out of the valid paths in the graph. 
We also fill the gaps while had been left from previous work. More specifically, our contributions are: 
\begin{itemize}
\item We derive a sharp upper bound on the number of intersections between departure curves (Section \ref{upperk}).
\item We prove some new properties of time-minimal paths which might be useful for other studies (Section \ref{section-preliminary}).
\item By synchronization of the discs from constants to polynomial, we study a more general class of motion planning problems.
\item We provide a solution which is simple and uses a reduction to a shortest path problem in graphs, which makes it also easier to implement.
\item We design a preprocessing algorithm which allow us to run $s$ and $d$ queries efficiently. Our query algorithm also decreases the existing time complexity for the different constant velocity discs problem from the previous work \cite{yi-thesis} by a factor of $n$. 
\end{itemize}

The remainder of the paper is organized as follows. Section \ref{section-preliminary} describes the preliminaries,  definitions and notations. 
 Section \ref{upperk} presents an upper bound for $k$ which affects the time and space efficiency of our algorithm. The main algorithm for the QTMP problem is described in Section \ref{section-basic-algorithm}. Finally, in Section \ref{section-conclusions} we conclude with open problems and future works.

\section{Preliminaries}\label{section-preliminary}
In this section, we formally define the QTMP problem.
We also introduce some preliminary concepts which are essential for the rest of this section.
  
Given are a point robot $r$ while moves with maximum velocity $\vr > 0$ and
a set of growing discs $\mathscr{D}=\{C_1, ..., C_n\}$.
A disc $C_i$ is given as a pair $(o_i, r_i(t))$, where $o_i$ is the center point and $r_i(t)$ is the radius of $C_i$ at time $t$. 
$\v{i}(t)$ is the velocity by which $C_i$'s radius grows at time $t$, where $\v{i}:T \rightarrow [0, \vr)$ is a polynomial function of degree $\dd$, such that $\v{i}(t) < \vr$. The following equation is readily obtained: $r_i(t) = \v{i}(t) t + R_i$, where $R_i$ is the initial radius of $C_i$ at time $t=0$. 
We denote by $C_i(t)$ the boundary of disc $C_i$ at time $t$. 
 A \textit{robot-path}  $\T$ between a start point $s$ and a destination point $d$ is given as function $\T:[T_s, T_d] \rightarrow {\rm I\!R}^2$, where $\T(t)$ returns the location of the robot at any time $t \in [T_s, T_d]$, provided that $\T(T_s)=s$ and $\T(T_d)=d$. We call $T_s$ \textit{departure time} and $T_d$ \textit{arrival time}.
 A path $\T$ \textit{intersects} disc $C_i$, if there exist a time $t \in [T_s, T_d]$, at which $\T(t)$ returns a point located in the interior of disc $C_i(t)$. 
  We call a path $\T$ \textit{valid} (or \textit{collision-free}) if it does not intersect any of the discs in $\mathscr{D}$; otherwise, $\T$ is \textit{invalid}. Without loss of generality, we assume that for all paths, $T_s = 0$ in our algorithms.

\begin{definition}\label{tmp}
A \textbf{time-minimal path} is defined as a valid path with arrival time $T_d$, where $T_d$ is minimized for all valid paths.
\end{definition}

The problem of finding a time-minimal path between two query locations $s$ and $d$ is referred to as 
\textit{Query Time-Minimal Path} (QTMP for short) problem. 
 In this section, we present an algorithm for the QTMP problem.

\begin{definition}\label{aligned-paths}
Two robot-paths $\T$ and $\T'$ are called \textbf{aligned paths}, if there exists a non-empty interval $[a, b]$, where $a < b$ and $\T(t) = \T'(t)$ for any $t \in [a,b]$. 
\end{definition}

 The following observation is an important property of time-minimal paths as derived in  \cite{shortest_path_in_growing_disc_referred}.

\begin{observation}\label{max-velocity}
On any time-minimal path, the point robot always moves with maximal velocity of $\vr$.
\end{observation}

\begin{figure}
\centering
\includegraphics[width=400pt]{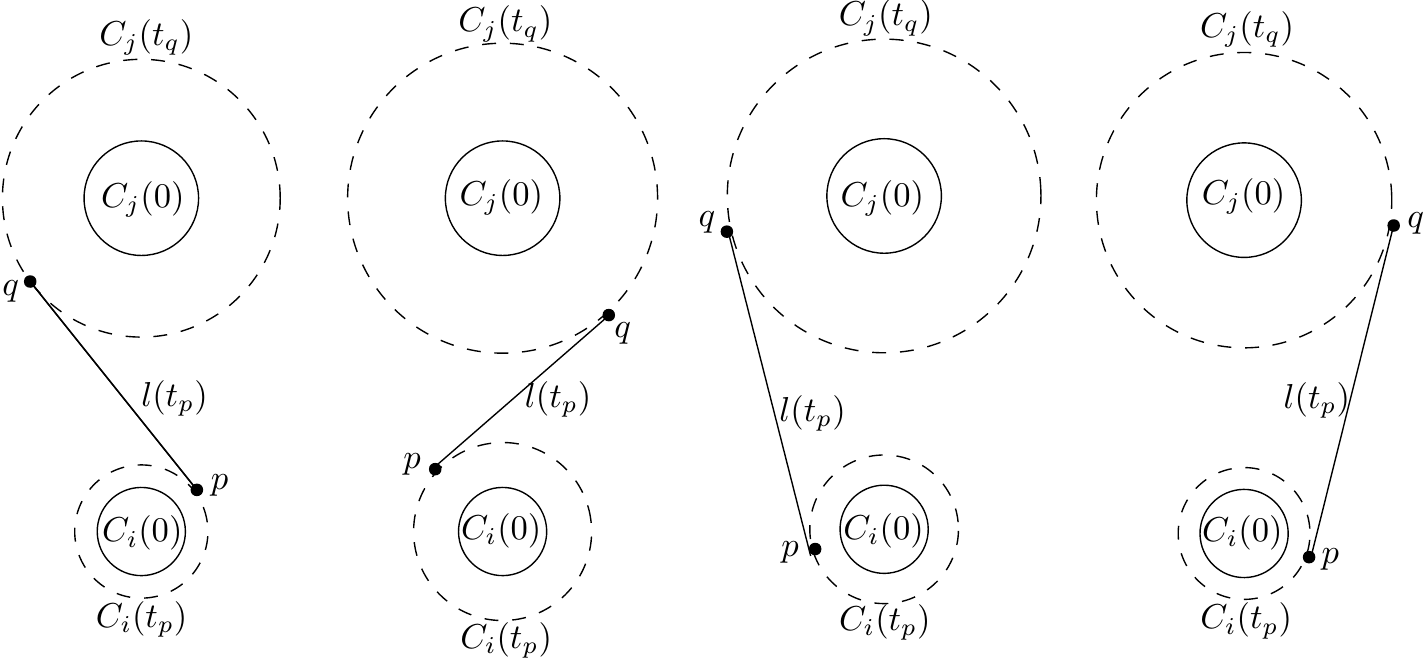}
\caption{This figure illustrates the four tangent paths between two discs $C_i$ and $C_j$. The left two figures show the inner tangent paths and the right two show the outer tangent paths. Note that $t_p$ is the departure time and $t_q$ is the arrival time where $t_q > t_p$.}\label{group-tangent-lines}
\end{figure}

Our approach to solve the QTMP problem involves examining the tangent lines between the
growing discs. A \textit{tangent line} $L$ between two discs $C_i$ and $C_j$, is a straight line which touches the boundaries of the two discs at exactly two points: $p \in C_i(t_p)$ and $q \in C_j(t_q)$. Refer to Figure \ref{group-tangent-lines},
we define a \textit{tangent path} $l(t_p, t_q)=\overline{pq}$ as a subsegment of $L$ between two points $p$ and $q$, where $|l(t_p, t_q)| = \vr (t_q - t_p)$. Note that $t_q$ is the arrival time at point $q$ if the robot departs from $p$ at time $t_p$. Refer to Observation \ref{max-velocity},  one can determine the arrival time $t_q$ when the departure time $t_p$ is given. So, for ease of notation, we denote a tangent $l(t_p, t_q)$ by $l(t_p)$ when no ambiguity arises.

As illustrated in Figure \ref{group-tangent-lines}, there are two types of tangents between any pair of discs: those for which corresponding discs are located on opposite sides of the tangent line, called \textit{inner} tangents, and those for which corresponding discs are located on same side of the tangent line, called \textit{outer} tangents. Following this definition, each pair of discs has two inner and two outer tangent paths.

Given two departure times $t$ and $t'$, where $t \not= t'$; due to the dynamic nature of the growing discs, it is observed that for any pair of discs $l(t) \not= l(t')$.
 The two (moving) endpoints of a tangent path, are called \textit{Steiner points}. The equation which calculated a Steiner point's motion is called a \textit{departure curve} and can be found in Section \ref{upperk}. 

 \begin{definition}\label{dc}
A \textbf{departure curve} $D_p:T \rightarrow {\rm I\!R}^2$ is a function specifying the location, $D_p(t)$ , of Steiner point $p$ at any time $t \in T$.
\end{definition}

\begin{figure}
\centering
\includegraphics[width=120pt]{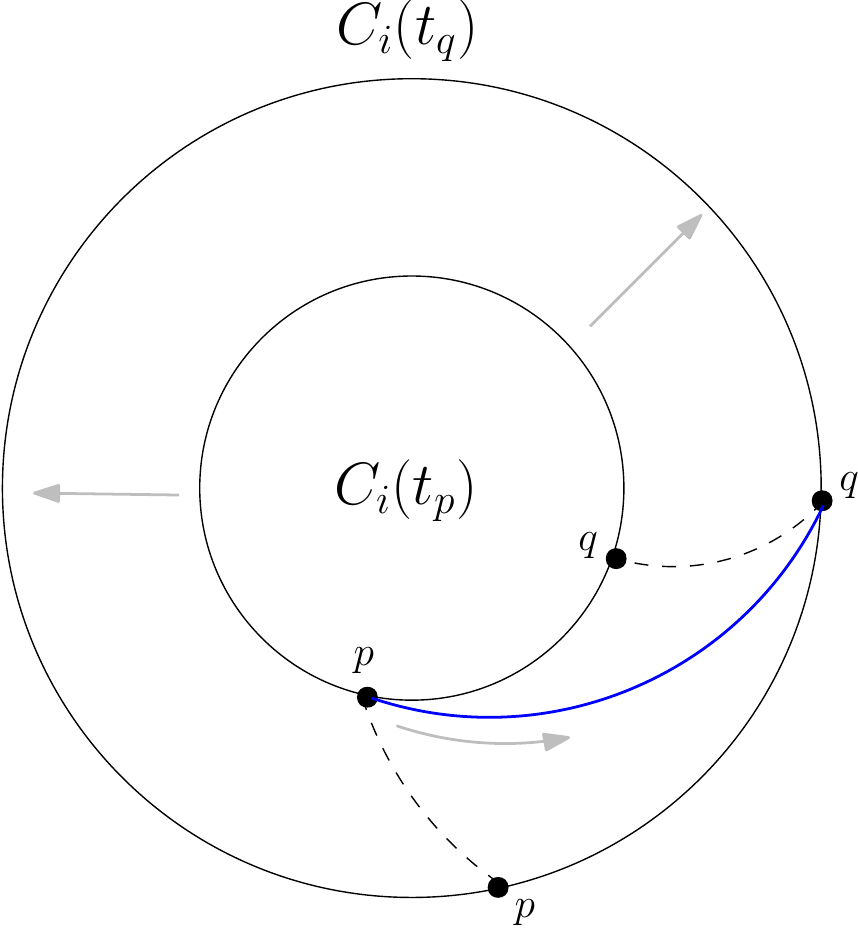}
\caption{The blue curve illustrates spiral path $s=\wideparen{pq}$. The point robots leaves Steiner point $p$ at time $t_p$, moves clock-wise on the boundary of $C_i$ and and arrives at Steiner point $q$ at time $t_q$.}\label{fig-spiral}
\end{figure}
 
%<<readd again>>>
Let $\steinerpoints = \cup_{1 \leq i \leq n} \{\steinerpoints_i\}$ where $\steinerpoints_i$ is a set of  Steiner points located on the boundary of disc $C_i$. Refer to Figure \ref{fig-spiral}, 
a \textit{spiral path} $s(t_p)$ is a robot-path where: (1) the robot departs from Steiner point $p \in \steinerpoints_i$ at time $t_p$, (2) moves either clockwise or counter-clockwise on the boundary of disc $C_i$ with speed $\vr$ and (3) stops when the next Steiner point $q \in \steinerpoints_i$ is visited. Like a tangent path, a spiral path is defined as a robot-path between two Steiner points. The main difference between these two path types is that the Steiner points of a tangent path are located on separate discs, while those of a spiral path are located on the same disc.

%<<read again>>>
Let us define $\ee$ as a set of all spiral and tangent paths. Each $pq \in \ee$ represents a robot-path where the point robot departs from $p$ and moves with the maximal velocity along the path, until it arrives at $q$. Note that for a path $pq \in \ee$ and a given departure time $\tau$, there exists a unique robot-path on the plane.  
The equation upon which these robot-paths are calculated can be found in Section \ref{upperk}. 
In the following, we prove that any time-minimal path consists only of paths in $\ee$.

\begin{lemma}\label{tang-spir}
A time-minimal path among growing discs is solely composed of tangent and/or spiral paths.
\end{lemma}\begin{proof} Let $\T$ be a time-minimal path.
If there exists a collision-free straight line path from the source to the destination, then $\T$ consists of the straight line between $s$ and $t$ which is a tangent path. Without loss of generality, assume that $\T$ is a time-minimal path which is not a (single) straight line. Let $\T'$ be a sub-path of $\T$ from point $p$ to point $q$ where: (1) $p \in \partial C_i$ and $q \in \partial C_j$, where $C_i \not= C_j$, (2) $\T'$ has no contact with the boundary of the discs except at two points $p$ and $q$, (3) no disc in $\mathscr{D}$ intersects $C_i$ or $C_j$ at points $p$ and $q$ respectively. Refer to Figure \ref{fig-spi-tan} (a) for an example. In the following, we prove: [a] $\T'$ is a straight line path, [b] $\T'$ is tangent to $C_i$ and $C_j$. 

\textbf{[a]} Assume (by contradiction) that $\T'$ is not a straight line segment. Then there is a point $o \in \T'$ such that: (1) no line segment containing $o$
is contained in $\T'$, (2) $o$ lies in the
interior of free space. Since $o$ does not touch the boundaries of the obstacles, there exists a ball with a positive radius $c$ centered at $o$, which is completely contained in free space\footnote{ A ball $c$ is located in free space, if there is a collision-free straight path from its center point to every point on its boundary. }. This is illustrated in  Figure \ref{fig-spi-tan} (a). Observe that $\lambda'$ intersects the boundary of $c$ in (at least) two points, denoted by  $x$ and $y$. Now, the sub-path of $\T'$ inside $c$ (which, by the assumption, is not a straight path), can be shortened by replacing it with the straight line path from $x$ to $y$ (see Figure \ref{fig-spi-tan} (a)). This contradicts the optimality of $\T'$.

\textbf{[b]} Let $\T'$ be a straight line segment $\overline{pq}$. For simplicity, we first assume $\v{i},\v{j} \ll \vr$, which implies that discs are static with respect to robot's motion. Now, assume by contradiction that $\T'$ is not tangent to $C_j$, refer to Figure \ref{fig-spi-tan} (b). 
Let $c$ be a ball with a positive radius and centered at point $q$ which intersects no disc other than $C_j$. Note that some part of the ball intersects the interior of $C_j$ and the other part, denoted by $\hat{c}$, is inside free space. As illustrated in Figure \ref{fig-spi-tan} (b),  $\T'$ intersects $c$ at point $x$. 
Let us denote by $\overline{xu}$ and $\overline{xu'}$ the two tangent lines to disc $C_j$ from point $x$. Without loss of generality, assume that $u'$ appears before $u$ in clockwise order around the boundary.  
Thus, it is observed that $(u', q, u)$ is the clockwise order of these three points on the boundary. 
In the following, we show that $\T'$ can be shortened inside $c$. Let $z$ and $z'$ be two points on the boundary of $C_j$ such that: (1) $(u', z', q, z, u)$ is the clockwise order of the points on the boundary, (2) $z$ and $z'$ are located inside $c$. This ensures that $q$ is located inside triangle $\bigtriangleup zxz'$. Assume, w.l.og, that the destination point is located outside of $\bigtriangleup zxz'$. Thus, $\T$ intersects the boundary of $\bigtriangleup zxz'$ at a point other than $x$. Denote this intersection point by $y$ (see Figure \ref{fig-spi-tan}). The two straight line paths $\overline{xz}$ and $\overline{xz'}$ are located inside $\hat{c}$ and are therefore collision-free. Thus, independent on whether $y$ lies on $\overline{xz}$ or $\overline{xz'}$, the path $\overline{xy}$ is collision-free. Finally, the sub-path of $\T$ between $x$ and $y$, can be shortened by the collision-free straight line path $\overline{xy}$. This is a contradiction to the optimality of $\T$. Now, in case  $\v{i},\v{j} \not\ll \vr$, we set the radius of $c$ small enough such that it does not intersect any disc other than $C_j$ until the robot has exited the triangle $\bigtriangleup zxz'$. Thus, paths $\overline{xz}$ and $\overline{xz'}$ are collision-free and therefore $\overline{xy}$ is collision-free. Again, $T$ can be shortened by $\overline{xy}$ which is a contradiction.

\begin{figure}[t]
\captionsetup[subfigure]{justification=centering}
\centering
    \begin{subfigure}{.4\linewidth}
    \includegraphics[width=2.3in]{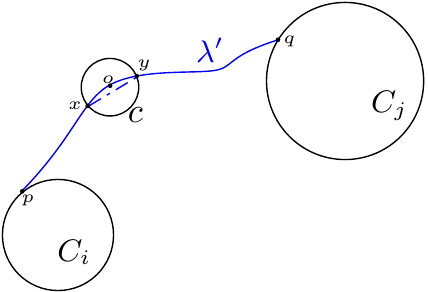}
    \centering
    \caption{}
  \end{subfigure}
    \begin{subfigure}{.4\linewidth}
    \includegraphics[width=2.3in]{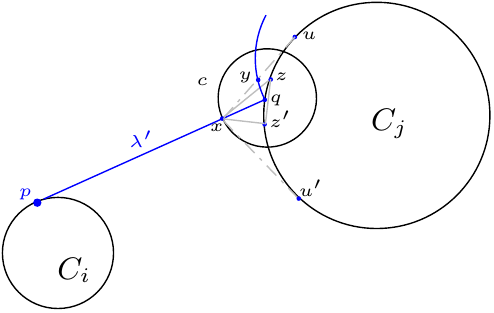}
    \caption{}
  \end{subfigure}
\caption{For the ease of demonstration, in these figures we assumed $\v{i},\v{j} \ll \vr$.}\label{fig-spi-tan}
\end{figure}

By [a] and [b], any maximal sub-path of $\T$ that is not in contact with the boundaries of the discs, is a tangent path. Let $\overline{pq}$ and $\overline{p'q'}$ be two consecutive tangent paths in $\T$. Then, we observe that the sub-path of $\T$ between $q$ and $p'$ is a path on the boundary of a disc in $\mathscr{D}$. By the definition of spiral paths, this sub-path consist of a sequence of spiral paths. Therefore, 
a time-minimal path is a sequence of tangent and/or spiral paths.
$\Box$
\end{proof}

This property of time-minimal paths allows us to construct a time-minimal path by finding a  sequence of paths in $\ee$ with minimum total length, as explained in the following. We construct a directed \textit{adjacency graph}, $G(V_s, E_s)$ as follows. With each Steiner point $v \in \steinerpoints$ we associate a unique vertex in $V_s$ and denote it by $\dot{v}$.
Then, with each path $\overrightarrow{uv} \in \ee$ we associate a unique edge $(\dot{u}, \dot{v})$ in $E_s$. 
Since each edge in $E_s$ is associated with a path in $\ee$, each path in graph $G$ is associated with a sequence of paths in $\ee$. In the following, we explain  how to construct a unique robot-path which is
 associated with a pair $(\pi, \tau)$, where $\pi$ is a path in $G$ and a $\tau$ is a departure time.

Note that for different departure times, the robot-path associated with edge $\overrightarrow{\dot{u}\dot{v}}$ may be different in shape and length.
Let $\robotpath:E_s \times T \rightarrow {\rm I\!R}^2$ be a function where $\robotpath(\overrightarrow{\dot{u}\dot{v}}, \tau)$ is the robot-path $\overrightarrow{uv}$, when the departure time is $\tau$. 
Because each robot-path is defined as a 2D curve on the plane, we can determine its length in constant time by looking up the corresponding equations in Section \ref{upperk}. For ease of notation, we use $e$ instead of $\overrightarrow{\dot{u} \dot{v}}$.
In the following, we denote the length of the path $\robotpath(e, \tau)$  by $|\robotpath(e, \tau)|$.

\begin{figure}
\centering
\includegraphics[width=5in]{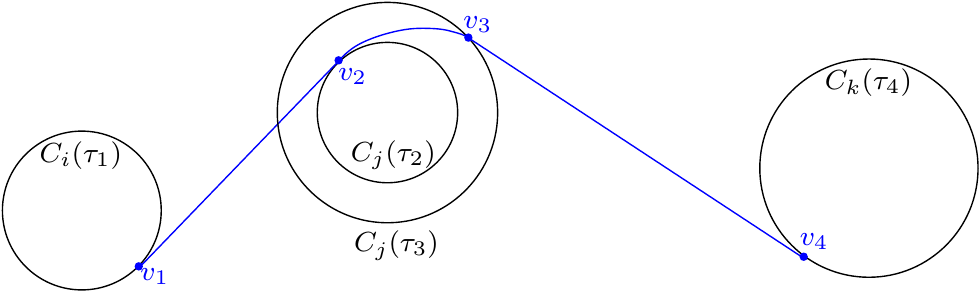}
\caption{This figure illustrates the robot-path $\robotpath(\pi(\dot{v_1}, \dot{v_4}), \tau_1) = (\robotpath((\dot{v_1}, \dot{v_2}),\tau_1), $ $\robotpath((\dot{v_2}, \dot{v_3}), \tau_2), \robotpath((\dot{v_{3}}, \dot{v_4}), \tau_3))$, where $\dot{v_1},\dot{v_2},\dot{v_3},\dot{v_4} \in V_s$ are vertices associated with Steiner points $v_1,v_2,v_3$ and $v_4$, respectively.}\label{fig-robot-path}
\end{figure}

 Let a simple path $\pi(v_1, v_h)$ between two vertices $v_1, v_h \in V_s$ be a finite connected sequence of edges in the adjacency graph, denoted by $\Big((v_1,v_2), (v_2,v_3), ..., (v_{h-1},v_h)\Big)$, such that no vertex appears more than once in the sequence. 
  Observe that, since for each pair $(e, \tau)$ there exists a unique robot-path, for each pair $(\pi, \tau)$ there is a unique robot-path as well.
See Figure \ref{fig-robot-path}. For a given path $\pi(v_1, v_h)$ in the graph and a departure time $\tau_1 \geq 0$, the associated robot-path $\robotpath(\pi(v_1, v_h), \tau_1)$ is defined as a sequence of robot-paths 
\begin{align*}\Big(\robotpath((v_1, v_2), \tau_1), \robotpath((v_2, v_3), \tau_2), ..., \robotpath((v_{h-1}, v_h), \tau_h)\Big),
\end{align*}
 where for $1 < i \leq h$ we have $\tau_i = \tau_{i-1} + {|\robotpath((v_{i-1}, v_{i}), \tau_{i-1})| \over \vr}$. 
 The length of this path is calculated as the summation of its sub-paths:
 \begin{align} 
 |\robotpath(\pi(v_1, v_h), \tau_1)| = \sum\limits_{i=1}^{h} {|\robotpath((v_{i-1}, v_i), \tau_i))|}
 \label{robot-path-size} \end{align}
Recall that an invalid robot-path is a path that intersects some disc $C_i$. In the following, we define two special cases for invalid robot-paths. 

\begin{definition}\label{validity}
An invalid robot-path $\T = \robotpath(e, \tau)$ is \textbf{blocked} if there exist a time $t \in (\tau, \tau + {|\robotpath(e, \tau)| \over \vr})$ and a disc $C_i$, such that $\T(t)$ returns a point located inside $C_i(t)$. For $t = \tau$ or $t = \tau + {|\robotpath(e, \tau)| \over \vr}$, if $\T(t)$ is inside $C_i(t)$, we say $\T$ is \textbf{dominated} by $C_i$.
\end{definition}

In Section \ref{section-the-algorithm}, we show how to determine if a query path is dominated or blocked. 
Given a point $x$ on the plane, we can determine a sequence of discs, intersecting $x$ over time. Since the discs are growing continuously, if disc $C_i$ intersects $x$ at time $t$, for any $t' > t$ we have $x \in C_i(t')$. Thus, point $x$ may not be part of a valid time-minimal path after time $t$. To identify the time interval at which we $x$ is located inside free space, we need to find the first disc in $\mathscr{D}$ which intersects $x$. To this end, we use the Voronoi diagram of the growing discs as follows. To identify if a path is dominated (refer to Lemma \ref{valid-time-d}), we first define a partitioning of the plane as follows.
\begin{definition}\label{voronoi}
The \textbf{Voronoi diagram of growing discs} is a partitioning of a plane into regions $\{\mathcal{H}_1,..., \mathcal{H}_n\}$, such that for any point $x \in \mathcal{H}_i$, $C_i$ is the first disc in $\mathscr{D}$ that intersects $x$, when the discs are growing.
\end{definition}

The Voronoi diagram of a set of growing discs is an \textit{additively multiplicatively weighted Voronoi diagram} \cite{Spatial-tessellations, yi-thesis}. We use a greedy algorithm proposed in \cite{yi-thesis} for computing this Voronoi diagram in $O(n^2 \log n)$ time. The degree of a Voronoi cell $\mathcal{H}_i$, denoted by $deg(\mathcal{H}_i)$, is the number of Voronoi edges of $\mathcal{H}_i$.

Define $\departurecurve_i = \cup_{p \in \steinerpoints_i} \{D_p\}$ as the set of departure curves originating from disc $C_i$. 
%EXPLAIN WHY
In Section \ref{section-basic-algorithm}, we compute the arrangement of these departure curves in $O(n\log n + k_i \log k_i)$, where $k_i$ is the number of intersections between the curves in $\departurecurve_i$.

\begin{definition} \label{arr-size-def}
Define the \textbf{arrangement size} as $k = \sum_{i=1}^n k_i$.
\end{definition}

The arrangement size is the total number of intersections in $n$ arrangements, corresponding to $n$ discs. In Section \ref{upperk}, we compute an upper bound on the arrangement size by deriving some properties for departure curves. 
%<<RELATION>>
In the following, we show an important property (Lemma \ref{bounds}) of  departure curves which will be used in Section \ref{section-basic-algorithm}.

\begin{figure}
\centering
\includegraphics[width=150pt]{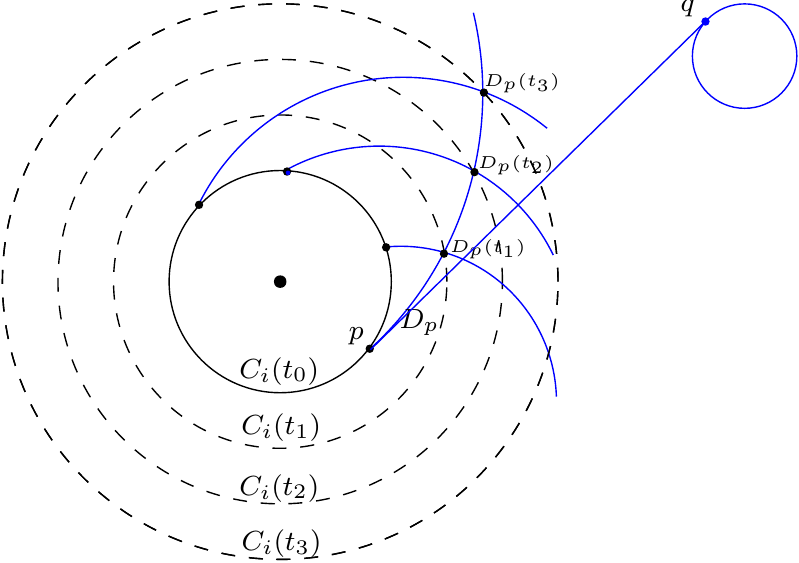}
\caption{The intersection sequence corresponding to to edge $e = (p, q)$ is $\alignset_e = \{x_{t_1},x_{t_2},x_{t_3}\}$.}\label{intersection-seq}
\end{figure}

Given a tangent path $\vec{e} = \vec{pq} \in \ee$, let $D_p \in \departurecurve_i$ be the departure curve of Steiner point $p$. Let $e \in G$ be the edge corresponding to $\vec{e}$.
We define an \textit{intersection sequence} $\alignset_e$ as a sequence of all  intersection points between $D_p$ and the departure curves in $\departurecurve_i \setminus \{D_p\}$, ordered by time of intersection. 
Refer to Figure \ref{intersection-seq} for an example. The intersection sequence in this example contains three points: $\alignset_e = \{x_{t_1},x_{t_2},x_{t_3}\}$. The three times $t_1$, $t_2$ and $t_3$ are when $D_p$ intersects other departure curves. Finally, $x_{t_1}$, $x_{t_2}$ and $x_{t_3}$ specify the locations of the intersection points at times $t_1$, $t_2$ and $t_3$, respectively.

Define a collection of all intersection sequences and call it \textit{intersection set}: $\alignset = \cup_{e \in E_s} \{\alignset_e\}$. In Section \ref{section-blocked}, we construct the intersection set by computing the arrangement of the departure curves.

\begin{observation}\label{blocked_by_c}
Given $a>0$ and $e \in E_s$, let $[a ,b]$ be a maximal time interval where for any $t \in [a, b]$, robot-path $\robotpath(e, \tau)$ is blocked by disc $C_h$. Then, $\robotpath(e, a)$ and $\robotpath(e, b)$ are tangents to $C_h$.
\end{observation}

\begin{proof} 
This is a consequence of the continuous movements of the robot and the discs. $\Box$
\end{proof}

\begin{lemma}\label{bounds}
For $a>0$ and $e \in E_s$, let $[a ,b]$ be a maximal time interval where for any $t \in [a, b]$, $\robotpath(e, \tau)$ is blocked by disc $C_i$. Then, we have $D_p(a), D_p(b) \in \alignset_e$.
\end{lemma} 

\begin{proof} 
Let $\robotpath(e, a)$ be a tangent path from $C_i$ to $C_j$. By Observation \ref{bounds}, $\robotpath(e, a)$ is also tangent to a third disc $C_h$. Thus, there is a tangent path between $C_i$ and $C_h$ which is \textit{aligned} (refer to Definition \ref{aligned-paths}) with $\robotpath(e, a)$ (see Figure \ref{fig-colinear} for an example). Let this tangent path be denoted by $\robotpath(e', a)$. This means that $\robotpath(e, a)$ and $\robotpath(e', a)$ are originating from the same point. Let $e = (p, q)$ and $e' = (p', q')$. Thus, we have $D_{p}(a) = D_{p'}(a)$, which by the definition, represents an intersection point between $D_{p}$ and $D_{p'}$. So, we have $D_{p}(a) \in \alignset_e$. 
Using a similar argument, we can prove that $D_{p}(b) \in \alignset_e$. 
$\Box$
\end{proof}

\begin{figure}
\centering
\includegraphics[width=250pt]{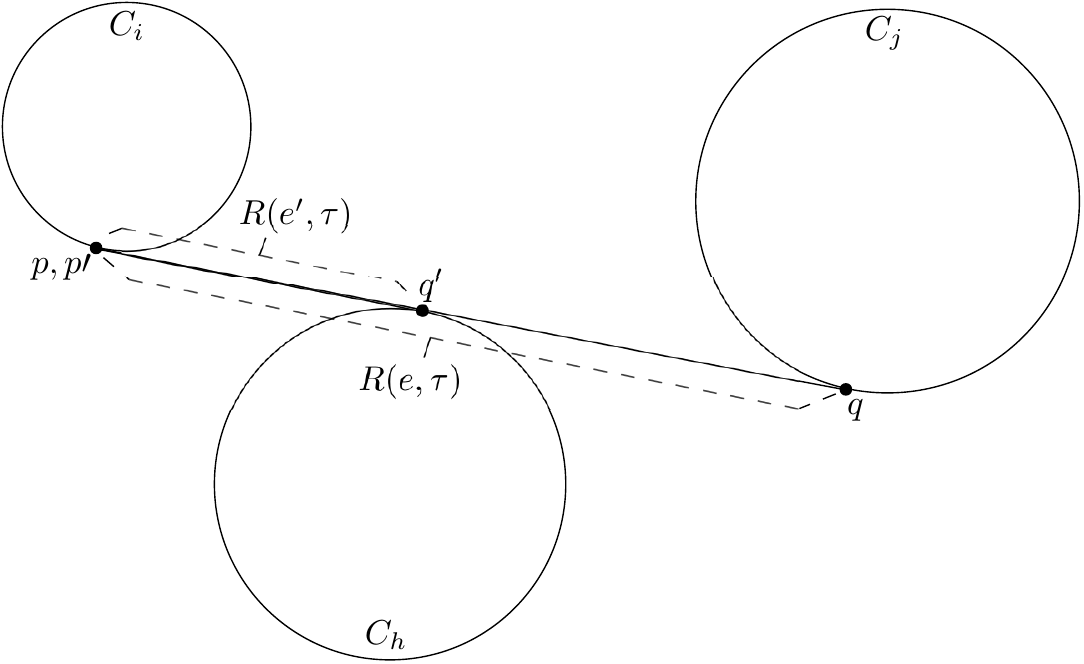}
\caption{This figure illustrates two aligned robot-paths $\robotpath(e, t)$ and $\robotpath(e', t)$ at time $t$ where $D_{\dot{p}}(t) = D_{\dot{p}'}(t)$.}\label{fig-colinear}
\end{figure}

For a given edge $e \in E_s$, a \textit{maximally blocked interval (MBI)} is a maximal time interval $I$ where for any $\tau \in I$, robot-path $\robotpath(e, \tau)$ is blocked.  Let us define a \textit{blocked sequence} $\invalidinterval_e=\{I_1, ..., I_h\}$ as all MBIs corresponding to edge $e$, sorted by their lower end points. Note that MBIs in $\invalidinterval_e$ are non-intersecting intervals. We define the \textit{blocked set} as $\invalidinterval = \cup_{e \in E_s}\{\invalidinterval_e\}$. In Section \ref{section-blocked}, we compute the blocked set as part of the preprocessing of Algorithm \ref{basic-algorithm}. In the Algorithm, we use $\invalidinterval$ to determine quickly if a query robot-path $\robotpath(e, \tau)$ is blocked in the algorithm.

\section{Upper Bound on $k$} \label{upperk}

\label{section-arrangment-size} 

In this section, we provide an upper bound on the number of intersection points between the departure curves (the arrangement size). Lemma \ref{bounds} proved that for each MBI, there are two distinct points in the intersection set associated with the two ends of the interval. In Section \ref{section-basic-algorithm}, we show the number of MBIs is not larger than the arrangement size. 
Since the validity of query robot-paths is carried out by searching on a list of MBIs (Section \ref{section-the-algorithm}), a bound on arrangement size directly impacts on the time complexity of our preprocessing and query algorithms. Thus, we are interested in finding an upper bound on $k$. Let us first derive some formulas for the departure curves. 

 The equations of the departure curves are most conveniently carried out in polar coordinates. Hence, we define a polar coordinate system as follows. A disc $C_i \in \mathscr{D}$ is denoted by a pair $(o_i, r_i(t))$, where $o_i$ is the center point and $r_i(t)$ is the radius of $C_i$ at time $t$. Locate the polar reference point at $o_i$. Now, a point $b$ in the plane is uniquely identified by a pair $(r, \theta)$, where $r$ is the Euclidean distance of point $b$ from the origin and $\theta$ is the angle between the positive horizontal axis and the line segment $\overline{o_ib}$ in counter clockwise direction. A departure curve is the location of a Steiner point over time.
We express the departure curve formula of a Steiner point $p \in \partial C_i$ by $(r_i(t), \theta_p(t))$ in polar coordinates, where $r_i(t)$ is the radius of disc $C_i$ at time $t$ and $\theta_p(t)$ is the polar angle of $p$ at time $t$. 
In the following, we derive the formula for $\theta_p(t)$.

As illustrated in Figure $\ref{fig-lr-tangent}$; let $l$ be a tangent path starting from Steiner point $p \in \partial C_i$ and arriving at Steiner point $q \in \partial C_j$, where $C_i \not= C_j$. Recall that for a given departure time $t$, $l(t)$ is a robot-path, when the robot leaves $p$ at time $t$.
Define a function $\delta_p : T \rightarrow [0, 2\pi)$, where $\delta_p(t)$ is the angle between the line segments $\overline{o_io_j}$ and $\overline{o_ip}$, at time $t$. Let $\theta_{ij}$ be the (constant) angle of segment $\overline{o_io_j}$ with respect to the horizontal line through $o_i$. Let $L$ be the line through $o_i$ and $o_j$ and directed from $o_i$ to $o_j$. Then, $L$ divides the plane into two regions: left half-plane and right half-plane.
We obtain:

\begin{align}\label{theta_delta}
        \theta_p(t) = \left\{  
          \begin{array}{ll}
                    \theta_{ij} + \delta_p(t) & \text{if $p$ is located in the left half-plane} \\
                    \theta_{ij} - \delta_p(t) & \text{if $p$ is located in the right half-plane}
                  \end{array}
                \right.
\end{align}

As in Figure \ref{fig-lr-tangent}, let $\overline{o_ic}$ be a parallel segment to $l(t)$ where $|l(t)| = |\overline{o_ic}|$. Now, consider the triangle $T = \triangle o_io_jc$. Since $l(t)$ is tangent to disc $C_j$, it is observed that $\overline{o_jc}$ is perpendicular to the segments $\overline{o_ic}$ and $l(t)$. Therefore, $\measuredangle c$ is a right angle and $T$ is a right triangle consequently. Now, we obtain the following equation:

\begin{align}
\label{delta}
&sin(\delta_p(t)) = {|l(t)| \over |\overline{o_io_j}|} \nonumber\\ 
&\delta_p(t) = arcsin({|l(t)| \over |\overline{o_io_j}|}) 
\end{align}

\begin{figure}[t]
\centering
\captionsetup[subfigure]{justification=centering}
    \begin{subfigure}{.5\linewidth}
    \includegraphics[width=2.5in]{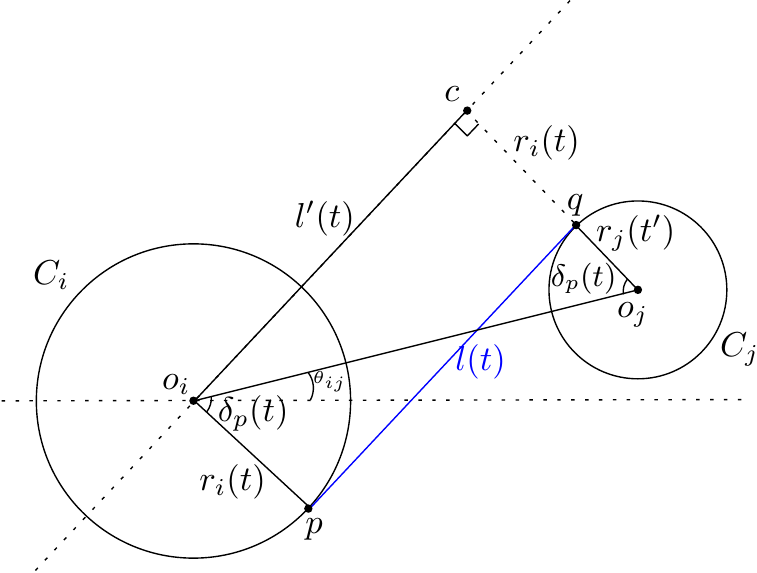}
    \caption{}
  \end{subfigure}
    \begin{subfigure}{.5\linewidth}
    \includegraphics[width=3in]{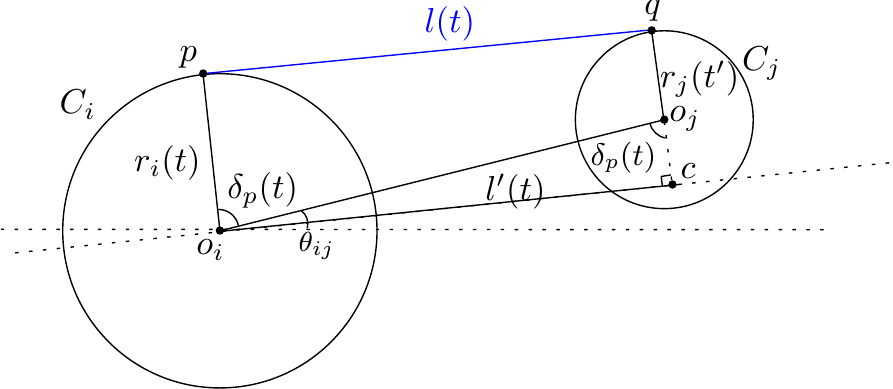}
    \caption{}
  \end{subfigure}

\caption{ (a) The inner tangent is denoted by $l(t)$, where we have $|\overline{o_jc}| = r_i(t) + r_j(t')$. (b)  The outer tangent is denoted by $l(t)$, where  $|\overline{o_jc}| = |r_i(t) - r_j(t')|$. 
In both (a) and (b) $\triangle o_ico_j$ is a right triangle.}\label{fig-lr-tangent}
\end{figure}

Then, using Pythagorean theorem we have: $|l(t)| = \sqrt{|\overline{o_io_j}|^2 - |\overline{o_jc}|^2}$. Note that when $l(t)$ is an inner tangent, we have $|\overline{o_jc}| = r_i(t) + r_j(t')$ and when $l(t)$ is an outer tangent, we have $|\overline{o_jc}| = |r_i(t) - r_j(t')|$, where $t' =  t+{|l(t)|\over \vr}$ (see Figure \ref{fig-lr-tangent}).
Therefore: 
\begin{align}\label{ell1}
        |l(t)| = \left\{  
          \begin{array}{ll}
                    \sqrt{|o_io_j|^2 - \Big(r_i(t)+r_j(t+{|l(t)|\over \vr})\Big)^2} & \qquad \text{if $\overline{pq}$ is an inner tangent},\\
                    \sqrt{|o_io_j|^2 - \Big(r_i(t)-r_j(t+{|l(t)|\over \vr})\Big)^2} & \qquad \text{if $\overline{pq}$ is an outer tangent}
                  \end{array}
                \right.
\end{align}

In the above equation, we replace $r_i(t)$ and $r_j(t')$ with $V_i(t)t + R_i$ and $V_j(t') t' + R_j$ respectively, where $t' = t+{|l(t)|\over \vr}$. Then, we isolate $|l(t)|$ as follows\footnote{In order to simplify Equation \ref{ell2}, we replaced $V_i(t)$, $V_j(t)$, $\vr$ and $|o_io_j|$ with $V_i$, $V_j$, $V_r$, and $L$, respectively.}.

{ 

\begin{align}\label{ell2}
 |l(t)| =  \left\{  
          \begin{array}{ll}
   {V_r V_j
   ( V_jt + R_j + V_it + R_i) + \sqrt{V_r^2\Big(- L^2V_j^2 + V_r^2 \Big(L^2  + (V_it + R_i + V_jt  + R_j)^2 \Big)\Big)} \over (V_r^2 - V_j^2)}  &\\
   \qquad\qquad\qquad\qquad\qquad\qquad\qquad\qquad\qquad\qquad\quad
    \text{if $\overline{pq}$ is an inner tangent,} &
   \\ & \\  &\\ &\\
   {V_r V_j
   ( V_jt + R_j - V_it - R_i)+ \sqrt{V_r^2\Big(- L^2V_j^2 + V_r^2 \Big(L^2  + (V_it  + R_i - V_jt - R_j)^2 \Big)\Big)}\over (V_r^2 - V_j^2)}
    &\\
   \qquad\qquad\qquad\qquad\qquad\qquad\qquad\qquad\qquad\qquad\quad
    \text{if $\overline{pq}$ is an outer tangent,} 
                    \end{array}
                  \right.
\end{align}

 Since $V_i(t)$ and $V_j(t)$ are polynomial functions of degree $\dd$, we can observe the following property in the above formula

\begin{property}\label{o(d)}
The equation for $|l(t)|$ can be expressed as a polynomial division ${f(t) + \sqrt{h(t)} \over g(t)}$ in which we can derive the following properties:

 \begin{itemize}
 \item(P1) $f(t)$ is a polynomials of degree $2\dd + 1$.
 \item(P1) $g(t)$ is a polynomials of degree $2\dd$.
 \item(P2) $h(t)$ is a polynomial of degree $4\dd + 2$.
 \item(P3) $g(t) > 0$.
 \end{itemize}
\end{property}

Now that we established the formula for $|l(t)|$, we use Equation \ref{delta} to explain the case when two departure curves intersect. 
Let $D_p$ and $D_{p'}$ be two departure curves where $l_{ij} = pq$ and $l_{ik} = p'q'$. By the definition, $D_p$ intersects $D_{p'}$ at time $t$, only if $D_p(t) = D_{p'}(t)$, and consequently $\theta_p(t) = \theta_{p'}(t)$. By Equations \ref{delta} and \ref{theta_delta}, we have:
\begin{align}
\label{inter-eq}
arcsin({|l_{ij}(t)| \over |\overline{o_io_j}|}) + \theta_{ij} = arcsin({|l_{ik}(t)| \over |\overline{o_io_k}|}) + \theta_{ik}
\end{align}
Note that $|\overline{o_io_j}|$, $|\overline{o_io_k}|$, $\theta_{ij}$ and $\theta_{ik}$ are constant values. To simplify the above equation, we assume $\theta_{ij} = \theta_{ik}$ and $|\overline{o_io_j}| = |\overline{o_io_k}|$. Since $arcsin$ is strictly increasing function, we have the following.
\begin{align*}
{|l_{ij}(t)|} = {|l_{ik}(t)|}
\end{align*}
By Property \ref{o(d)}, both $|l_{ij}(t)|$ and $|l_{ik}(t)|$ are divisions of the form ${f(t) + \sqrt{h(t)} \over g(t)}$ where $f(t)$, $g(t)$ and $h(t)$ are polynomials. If $|l_{ij}(t)| = {f_1(t) + \sqrt{h_1(t)} \over g_1(t)}$ and $|l_{ik}(t)| = {f_2(t) + \sqrt{h_2(t)} \over g_2(t)}$, then
\begin{align*}
{f_1(t)g_2(t) + \sqrt{h_1(t)}g_2(t) - f_2(t)g_1(t) - \sqrt{h_2(t)}g_1(t)} = 0
\end{align*}
where  $g_2(t),g_2(t) > 0$. We eliminate the radicals of the formula by squaring the equation twice, which results in the following:
\begin{align*}
\Bigg({{{\Big(f_1(t)g_2(t)  - f_2(t)g_1(t)\Big)^2} - g_1^2(t)h_2(t)- g_2^2(t)h_1(t)} \over 2g_1(t)g_2(t)}\Bigg)^2 - h_1(t)h_2(t) = 0
\end{align*}
 Observe that this is a polynomial of degree $16\dd + 8$. So, there are at most $16\dd + 8$ distinct values for $t$ such that ${|l_{ij}(t)|} = {|l_{ik}(t)|}$. This implies that the two departure curves $D_p(t)$ and $D_{p'}(t)$ may intersect $O(\dd)$ times. Since 
there are $n$ departure curves in $\departurecurve_i$, the number of intersections in the arrangements of $\departurecurve_i$ is of order $O(n^2 \dd)$. Thus, the total number of intersection points in all arrangements (i.e., the arrangement size) is $O(n^3 \dd)$. 

\begin{corollary}\label{maxip}
The arrangement size is upper bounded by $O(n^3 \dd)$.
\end{corollary}

\section{Time-Minimal Path Algorithm} \label{section-basic-algorithm}

In this section, we present the algorithm to answer time-minimal path queries among growing discs. We first discuss the preprocessing algorithm which runs in $O(n^2 \log n + k \log k)$ time, where $n$ is the number of discs in $\mathscr{D}$ and $k$ is the arrangement size.
For a pair of  query points $s$ and $d$, our algorithm  computes a time-minimal path from $s$ to $d$ among the growing discs in $O(n^2 \log (kn))$ time.
The query time could substantially reduce the running time since $k \in O(n^3\dd)$. We achieve our result by constructing data structures, so that we determine if query path is valid efficiently. 
 We
establish this claim in Lemmas \ref{valid-time-b} and \ref{valid-time-d}. 

\subsection{Preprocessing} \label{section-blocked}
In this section, we preprocess the departure curves using a data structure, designed such that given an edge $e \in E_s$ and its associated query robot-path $\T$, we can quickly determine if $\T$ is blocked. In Section \ref{section-preliminary}, we called this data structure the blocked set, which is the collection of all blocked sequences.

 We construct the blocked set $\invalidinterval$ in two steps. First, we construct the intersection set $\alignset$, by computing the arrangement of departure curves. Second, we run Algorithm \ref{preprocess}, whose inputs are the adjacency graph $G$ and the intersection set $\alignset$, and its output is $\invalidinterval$. 
 % Using the blocked set we can verify if a given query path is blocked, as will be shown later in Lemma \ref{valid-time-b}.

To compute the arrangement of departure curves, we employ the Bentley-Ottmann sweepline algorithm \cite{Bentley-Ottmann}. The main idea of their algorithm is to move a vertical sweep-line from left to right across the plane, intersecting the input objects sequentially. The input objects of this algorithm must satisfy the following three properties: 
\begin{itemize}
\item 
(P1) Any vertical line intersects
each object exactly once.
\item (P2) For any pair of objects intersecting the same
vertical line we can determine which is above the other one at constant cost.
\item 
(P3) Given two objects, it is possible to compute their
leftmost intersection point (if they have any), after some fixed vertical line.
\end{itemize} 
In the following lemma we propose a modified version of the Bentley-Ottmann algorithm to compute the arrangement of departure curves of disc $C_i$. This is an important part the preprocessing computations (Line 3 of Algorithm \ref{pre-algorithm}). The idea is to use a circle sweep instead of the vertical line-sweep. A \textit{circle sweep} $\zeta$ is defined as a circle which sweeps the plane by continuously growing  starting from $C_i$'s center point (see Figure \ref{fig-departure-arrangement} for an example).

\begin{figure}
\centering
\includegraphics[width=200pt]{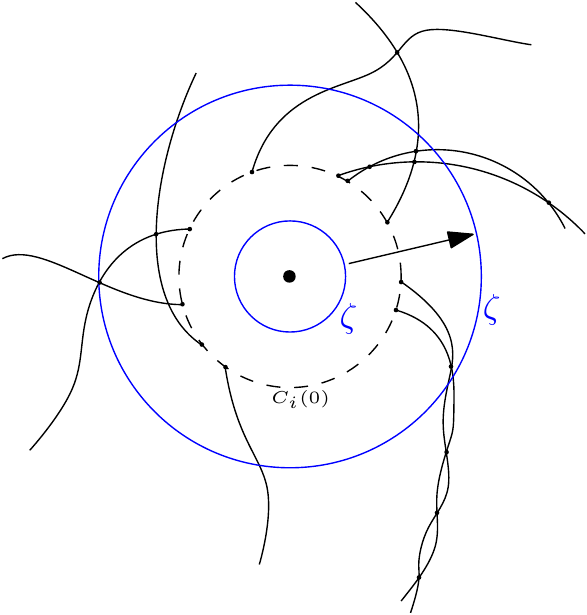}
\caption{Circle sweep $\zeta$ grows outwards from the center point of disc $C_i$ and sweeps the departure curves.}\label{fig-departure-arrangement}
\end{figure}

\begin{lemma}\label{arrangment}
Intersection set $\alignset$ can be computed in $O(n^2 \log n + k \log k)$ time, where $k$ is the arrangement size  (the total number of intersection points in all arrangements).
\end{lemma} 

\begin{proof} 
Let $\mathcal{D}_i$ be a set of departure curves originating from disc $C_i$. We first compute the arrangement of departure curves in $\mathcal{D}_i$. In order to achieve this, we adopt the Bentley-Ottmann paradigm. We prove that the departure curves meet the three properties of input objects in Bentley-Ottmann algorithm, when we use a circle sweep instead of a vertical sweep:
\begin{itemize}
\item (P1) A fixed circle sweep intersects each departure curve in $\mathcal{D}_i$ exactly once.
\item (P2) Let $D_p, D_q \in \mathcal{D}_i$ be two departure curves intersecting a fixed circle sweep $\overline{\zeta}$. Let $D_p$ and $D_q$ intersect $\overline{\zeta}$ at points  $\overline{p}$ and $\overline{q}$, respectively. Then, we can determine the clock-wise order in which $\overline{p}$ and $\overline{q}$ appear on the boundary of the circle sweep.
\item (P3) Given two departure curves in $\mathcal{D}_i$, it is possible to compute their
closest intersection point (if they have any) to some fixed circle sweep.
\end{itemize}

We prove Property 1 by contradiction. Assume a fixed circle sweep intersects a departure curve $D_p$ at two points $D_p(t_1)$ and $D_p(t_2)$ where $t_1 < t_2$. Observe that  $D_p(t_1)$ and  $D_p(t_2)$ lie at an equal distance from the center point. Thus, the radius of $C_i$ is equal at two times $t_1$ and $t_2$: $r_i(t_1) = r_i(t_2)$. This is a contradiction since we assumed that the discs are growing constantly.

Property 2 is established straightforwardly by applying Equation \ref{theta_delta}, in which, for two given departure curves $D_p, D_q \in \mathcal{D}_i$ and a time $t$, the polar angles $\theta_p(t)$ and $\theta_q(t)$ are provided. Thus, the clockwise order of the intersection points on the boundary of $\zeta$ can be found using the polar angles.

In order to prove Property 3, we first use Equation \ref{inter-eq} to compute the intersection point(s) of two given departure curves $D_p, D_q \in \mathcal{D}_i$. Then, we sort them based on their distances to $C_i$'s center in $O(k' \log k')$, where $k'$ is the number of intersections. Since we repeat this procedure for every pair of intersecting departure curves, the time complexity is $O(k_i \log k_i)$, with $k_i$ being the number of intersections among the departure curves in $\mathcal{D}_i$.

Following the analysis of the Bentley-Ottmann plane sweep paradigm \cite{Bentley-Ottmann}, our algorithm runs in $O(n \log n + k_i \log k_i)$ time. We run this algorithm once for each disc to compute $n$ arrangements corresponding to $n$ discs. Thus, the total running time of computing all arrangements is $O(n^2 \log n + k \log k)$, where $k = \sum_{i=1}^n k_i$. After finding the arrangements, sorting the intersecting points corresponding to each departure curve is straightforward. Thus, the intersection set $\alignset$ can be computed in $O(n^2 \log n + k \log k)$ time.
$\Box$
\end{proof}

Recall that each blocked sequence $\invalidinterval_e \in \invalidinterval$, contains all the maximal intervals $[a, b]$ such that for any $\tau \in [a, b]$, robot-path $\robotpath(e, \tau)$ is blocked. Now, simply observe that, by performing a binary search for $\tau$ in $\invalidinterval_e$, we can determine in logarithmic time if $\robotpath(e, \tau)$ is blocked. Next, we present Algorithm \ref{preprocess}, which computes the blocked set $\invalidinterval$.
 
 Algorithm \ref{preprocess} starts by creating a list of discs, denoted by $\mathcal{C}$, which intersects path $\robotpath(e, 0)$. Then, we update this list throughout the algorithm, so that, for a given departure time $\overline{t}$ (in line \ref{popx}), set $\mathcal{C}$ contains all discs intersecting $\robotpath(e, \overline{t})$. Trivially, if $\mathcal{C}$ is not empty at time $\overline{t}$, then $\robotpath(e, \overline{t})$ is blocked, and vice-versa. This algorithm is explained in more detail in Lemma \ref{ppp}.

\renewcommand{\baselinestretch}{1.3}

\begin{algorithm}[H]
\caption{$BlockedSet(G, \alignset)$}
\label{preprocess}
\begin{description}
\item {\bf Input}: The adjacency graph $G(V_s, E_s)$ and intersection sets $\alignset = \cup_{e \in \steinerpoints}\{\alignset_p\}$.
\item {\bf Output}: Blocked set $\invalidinterval = \cup_{e \in E_s} \{\invalidinterval_e\}$
\end{description}

\begin{algorithmic}[1] 
\STATE initialize $\invalidinterval = \{\emptyset\}$
 \FOR { each edge $e=(\dot{p}, \dot{q})  \in E_s$}\label{for} \par
 \COMMENT{ $\dot{p}, \dot{q} \in V_s$ are vertices associated to Steiner points $p$ and $q$. }
  \STATE initialize $\invalidinterval_e = \{\emptyset\}$, $low = 0$ and $high = 0$
  \STATE \label{initD} let $\mathcal{C}$ be a list of discs that intersect  $\robotpath(e, 0)$ \par
\COMMENT{ $\robotpath(e, 0)$ is the robot-path of $\vec{e}$, when the departure time is 0.}
  \WHILE {$j \leq h$}\label{while} \COMMENT{ $\alignset_e = \{D_p(t_1), ..., D_p(t_h)\}$ is the intersection sequence of $p$.}
    \STATE \label{popx} pop point $D_p(t_j)$ from $\alignset_e$ \COMMENT{Note that $t_j < t_{j+1}$} 
    \STATE let $D_p(t_j)$ be an intersection point between $D_{p}$ and $D_{p'}$
    % \STATE let $l_{iu} = \overline{p'q'}$ where $p' = CS(x)$
    \STATE let $e'=(\dot{p}', \dot{q}')$ 
 \COMMENT{ $\dot{p}', \dot{q}' \in V_s$ are vertices associated to Steiner points $p'$ and $q'$. } 
% \algstore{myalg}
% \end{algorithmic}
% \end{algorithm}

% \begin{algorithm}[H]                    
% \begin{algorithmic} [1]
% \algrestore{myalg}

    \IF {$|\robotpath(e', \overline{t})| < |\robotpath(e, \overline{t})|$} \label{robotif}
 \COMMENT{ $\robotpath(e, \overline{t})$ and $\robotpath(e', \overline{t})$ are matching paths. }
      \STATE \label{cu} let $q' \in \partial C_u$ 
      %  At this point, $C_u$ is a tangent disc to robot-path $\robotpath(e, \overline{t})$. 
      \COMMENT{ If $C_u \in \mathcal{C}$, then for some $\epsilon > 0$, $\robotpath(e, \overline{t} - \epsilon )$ intersects $C_u$. }
      \IF {$C_u \in \mathcal{C}$} 
        \STATE \label{removec}  remove $C_u$ from $\mathcal{C}$
        \STATE \label{high} $high = \overline{t}$
        \IF {$\mathcal{C}$ is empty} 
        \par \COMMENT{ If $\mathcal{C}$ is empty, then for some $\epsilon > 0$, $\robotpath(e, \overline{t} + \epsilon )$ is not blocked.}      
          \STATE \label{lowup} insert $[low, high]$ to $\invalidinterval_e$
        % \par
        % \COMMENT{Because the robot-path is not blocked $'high'$ is the right end point of the blocked interval.}
        \ENDIF
      \ELSE   
      %  If $C_u \not\in \mathcal{C}$, then $C_u$ will be  intersecting $\robotpath(e, \overline{t})$. 
      
        \IF {$\mathcal{C}$ is empty} 
        \par \COMMENT{ If $\mathcal{C}$ is empty, then for some $\epsilon > 0$, $\robotpath(e, \overline{t} - \epsilon )$ was not blocked.}       
          \STATE \label{low} $low = \overline{t}$ \label{setlow}
        \ENDIF

        \STATE \label{end-while} add $C_u$ to $\mathcal{C}$ \label{add_i}
      \ENDIF
    \ENDIF 
  \ENDWHILE
  \STATE \label{addc} add $\invalidinterval_e$ to $\invalidinterval$.
 \ENDFOR
\STATE \bf{return} $\invalidinterval$
\end{algorithmic}
\end{algorithm}

\begin{lemma}\label{ppp}
Algorithm \ref{preprocess} computes the blocked set correctly.
\end{lemma}

\begin{proof} 
Before the main loop (lines \ref{while}-\ref{end-while}) of the algorithm, in line \ref{initD}, $\mathcal{C}$ is defined as a list of discs that intersect $\robotpath(e, \tau = 0)$. Since $\robotpath(e, \tau > 0)$ may intersect a different set of discs than $\robotpath(e, \tau = 0)$, the algorithm updates $\mathcal{C}$ in each iteration accordingly. So, we first prove the following loop invariant for the main loop:

\begin{quotation}\noindent
\textit{At the beginning of each iteration, set $\mathcal{C}$ contains all discs that intersect the robot-path $\robotpath(e, \overline{t})$.}
\end{quotation}

The invariant holds the first time line \ref{while} has already been executed, since at time $t=0$, set $\mathcal{C}$ is already calculated (in line \ref{initD}) and contains all discs that intersect $\robotpath(e, \tau = 0)$. Now, assume the invariant holds for the $j^{th}$ iteration. By Observation \ref{blocked_by_c}, for a maximal time interval $[a, b]$ where $\robotpath(e, \tau \in [a, b])$ intersects a disc $C_i$, $\robotpath(e, a)$ and $\robotpath(e, b)$ are tangents to disc $C_i$. 
Thus, to update set $\mathcal{C}$, it suffices to look at the events when $\robotpath(e, \tau)$ is tangent to a disc in $\mathscr{D}$. 
In line \ref{cu}, the disc which is tangent to path $\robotpath(e, \tau)$ is identified. Next, if $C_u \in \mathcal{C}$, in line \ref{removec} $C_u$ is removed from $\mathcal{C}$.  If $C_u \not\in \mathcal{C}$, then $C_u$ is added  to  $\mathcal{C}$ in line \ref{addc}. Thus, at the beginning of the ${(j+1)}^{th}$ iteration $\mathcal{C}$ contains the discs in $\mathscr{D}$ that intersect $\robotpath(e, \tau = \overline{t})$.
% So, w.l.o.g., from this point forward, for any given $\tau$, we assume $\mathcal{C}$ contains all the discs intersecting $\robotpath(e, \tau)$. 

% Let $\invalidinterval_e \in \invalidinterval$, where $\invalidinterval$ is the output of Algorithm \ref{preprocess}.

Next, let  $[a, b]$ be a maximally blocked interval corresponding to edge $e$. By Lemma \ref{bounds}, there exist $D_p(t_{m_1}=a), D_p(t_{m_2}=b) \in \alignset_e$. In the following, we show that interval $[low, high]$ is added to $\invalidinterval_e$ when $low = a$ and $high= b$.

Assume $a > 0$. Because $[a, b]$ is a maximal blocked interval, there exist an interval $[a - \epsilon, a)$, for some $\epsilon > 0$, such that for any $\overline{\tau} \in [a - \epsilon, a)$, robot-path $\robotpath(e, \overline{\tau})$ is not blocked. Now, consider the $m_1^{th}$ iteration, when point $D_p(t_{m_1}=a)$ is removed from the list $\alignset_e$ (line \ref{popx}). At the beginning of this iteration, because $[a - \epsilon, a)$ is a non-blocked interval, by the loop invariant,  $\mathcal{C}$ is empty. 
Thus, line \ref{setlow} is executed and we have $low = a$.
In case $a = 0$, $\robotpath(e, 0)$ intersects some disc(s) and $\mathcal{C}$ is not empty at the beginning of the iterations. Thus, we have $low = a = 0$ in any loop before the $m_1^{th}$ iteration.

On the other hand, since $[a, b]$ is a maximally blocked interval, there exist an $\epsilon' > 0$ such that: (a) for any $\overline{\tau} \in [b - \epsilon', b]$, robot-path $\robotpath(e, \overline{\tau})$ intersects some disc $C_u$, (b) for any $\overline{\tau} \in (b, b + \epsilon']$, robot-path $\robotpath(e, \overline{\tau})$ is not blocked. First, since $C_u$ is tangent to robot-path $\robotpath(e, b)$, it is easily seen that in line \ref{robotif} we have $|\robotpath(e', b)| < |\robotpath(e,  b)|$. Now consider the ${m_2}^{th}$ iteration when $D_p(t_{m'} = b)$ is popped out of $\alignset_e$ (line \ref{popx}). By statement (a), $C_u \in \mathcal{C}$ and by (b), $\mathcal{C}$ becomes empty after removing $C_u$ (line \ref{removec}). So, line \label{high} is executed and $high = b$.

Since $[a, b]$ is a blocked interval, $\mathcal{C}$ is non-empty at any iteration
between the $m_1^{th}$ and  ${m_2}^{th}$ iterations. Therefore, when $high= b$ we also have $low = a$. In this iteration, line \ref{add_i} is executed and $[low, high]$ is added to $\invalidinterval_e$.
$\Box$

%<<sorte>>
%<<the other side of lemma>>
\end{proof}

\begin{lemma}\label{ppp-time}
Algorithm \ref{preprocess} runs in $O(n^2 + k)$ time.
\end{lemma}

\begin{proof}  
The algorithm is iterated until all $n^2$ tangent paths have been processed. In each iteration, the while loop, starting at 5, is executed $|\alignset_p|$ times. Since the total number of points in all intersection sequences (i.e. the arrangement size) is $k$, the while loop executes $k$ times. Since $\alignset_p$ is given as a sorted list, finding $D_p(\overline{t})$ with the minimum time $\overline{t}$ in line \ref{popx} can be done in constant time. Thus, the time complexity of the algorithm is $O(n^2 + k)$.
$\Box$
\end{proof}

In Algorithm \ref{pre-algorithm}, we summarize the steps of the preprocessing. As explained in Section \ref{section-preliminary}, we compute the Steiner points, tangent lines and adjacency graph in $O(n^2)$ time. By Lemma \ref{arrangment}, the intersection set $\alignset$ can be computed in $O(n^2 \log n + k \log k)$. By Lemma \ref{ppp-time}, the blocked set $\invalidinterval$ is computed in $O(n^2 + k \log k)$. Finally, the Voronoi diagram $\mathcal{H}$ (refer to Definition \ref{voronoi}), can be computed in $O(n^2 \log n)$ time using the algorithm proposed in \cite{yi-thesis}.

\begin{algorithm}[H]
\caption{Preprocessing} \label{pre-algorithm}

\begin{description}
\item {\bf Input}: A set of growing discs $\mathscr{D}$. 
\item {\bf Output}: Blocked set $\invalidinterval$ and Voronoi Diagram $\mathcal{H}$.
\end{description}
\begin{algorithmic}[1]
\STATE find the tangent lines and Steiner points \COMMENT{$O(n^2)$}
\STATE construct the adjacency graph $G$ \COMMENT{$O(n^2)$}
\STATE compute the intersection set $\alignset$ \COMMENT{$O(n^2 \log n + k \log k)$}
\STATE $\invalidinterval = BlockedSet(G, \alignset)$ \COMMENT{$O(n^2 + k \log k)$}
\STATE compute the Voronoi Diagram $\mathcal{H}$ \COMMENT{$O(n^2 \log n)$}
\end{algorithmic}
\end{algorithm}

\subsection{Query Time-Minimal Path Algorithm} \label{section-the-algorithm}
In this section, we present an algorithm to compute a valid time-minimal path among growing discs from a query source point $s$ to a query destination point $d$. We assume that the preprocessing is done as per Algorithm \ref{pre-algorithm}. First, we determine if a given query path $\robotpath(e, \tau)$ is valid. By Definition \ref{validity}, an invalid robot-path is either dominated or blocked. We show how to determine the validity of a query robot-path in Lemmas \ref{valid-time-b} and \ref{valid-time-d}. Then, we define a weight function for which the inputs are $e$ and $\tau$, and the output is the length of the path $\robotpath(e, \tau)$ if it is valid, and $\infty$ otherwise. Next, we run Dijkstra's algorithm where it measures the distance between two vertices using our customized weight function (refer to Algorithm \ref{basic-algorithm}). In Lemma \ref{dijk-proof} we prove that our \textsc{\char13}modified\textsc{\char13}  Dijkstra's algorithm produces a time-minimal path.

\begin{lemma}\label{valid-time-b}
For a given edge $e \in E_s$ and a departure time $\tau$, we can determine if $\robotpath(e, \tau)$ is blocked in $O(\log k)$ time.
\end{lemma}

\begin{proof} 
The elements in $\invalidinterval_e$ are non-intersecting time intervals sorted by their lower endpoints. Thus, by performing a binary search we can determine in $O(\log k)$ time if $\tau$ belongs to an interval in $\invalidinterval_e$ or not. $\Box$
\end{proof}

In the following, we show how to determine if a path is dominated using the Voronoi diagram of growing discs. Recall that $deg(\mathcal{H}_i)$ is the number of edges of Voronoi cell $\mathcal{H}_i$

\begin{lemma}\label{valid-time-d}
 Let $e \in E_s$ be the associated edge of path $pq \in \ee$, where $p \in \partial C_i$ and $q \in \partial C_j$. For a given departure time $\tau$, we can determine if $\robotpath(e, \tau)$ is dominated in $O(deg(\mathcal{H}_i) + deg(\mathcal{H}_j))$ time.
\end{lemma}

\begin{proof}  
Let $\robotpath(e, \tau)$ be a robot-path dominated by disc $C_m$.
By Definition \ref{validity}, we have $D_p(\tau) \in C_m(\tau)$ (or $D_q(\tau + {|\robotpath(e, \tau)| \over \vr}) \in C_m(\tau + {|\robotpath(e, \tau)| \over \vr})$). Let $\mathcal{H}_i$ be the Voronoi cell corresponding to disc $C_i$. Since $D_p(\tau)$ is located inside $C_m(\tau)$ and on the boundary of $C_i(\tau)$, $C_m$ intersects $D_p(\tau)$ before $C_i$. Thus, by Definition \ref{voronoi}, we have $D_p(\tau) \not\in \mathcal{H}_i$. Using a similar argument, $\robotpath(e, \tau)$ is dominated also when $D_q(\tau + {|\robotpath(e, \tau)| \over \vr}) \not\in \mathcal{H}_j$. 
Thus, to check if $\robotpath(e, \tau)$ is dominated, it suffices to  perform a point location algorithm \cite{point-location} on Voronoi cells $\mathcal{H}_i$ and $\mathcal{H}_j$ to determine if $D_p(\tau) \not\in \mathcal{H}_i$ or $D_q(\tau + {|\robotpath(e, \tau)| \over \vr}) \not\in \mathcal{H}_j$. The point location procedure runs in $O(deg(\mathcal{H}_i) + deg(\mathcal{H}_j))$ \cite{point-location}, where $deg(\mathcal{H}_i)$ and $deg(\mathcal{H}_j)$ are the number of edges of Voronoi cells $\mathcal{H}_i$ and $\mathcal{H}_j$.
$\Box$
\end{proof}

To simplify the description of our QTMP algorithm, we let the two query points $s$ and $d$ be two discs with zero radius and zero velocities and add them to $\mathscr{D}$. Denote the new set by $\overline{\mathscr{D}} = \mathscr{D} \cup \{s, d\}$. Then, the problem is to find a query time-minimal path from disc $s$ to disc $d$. As discussed in Section \ref{section-preliminary}, we can calculate the tangent paths origination from $s$ or ending at $d$, in linear time. Denote these tangent paths by $\overline{\ee}$.
Let $\overline{G}=(V_s \cup \overline{V_s}, E_s \cup \overline{E_s})$ be the adjacency graph corresponding to set $\overline{\mathscr{D}}$, where $\overline{E_s}$ and $\overline{V_s}$ are the (new) edges and the vertices associated with the tangent paths in $\overline{\ee}$.

From Lemmas \ref{valid-time-b} and \ref{valid-time-d}, it follows that for a given edge $e \in E_s$ and departure time $\tau$, we can determine in $O(\log k + deg(\mathcal{H}_i) + deg(\mathcal{H}_j))$ time if robot-path $\robotpath(e, \tau)$ is valid . Also note that for each edge $\overline{e} \in \overline{E_s}$ the validity of path $\robotpath(e, \tau)$ can be determined by simply comparing the robot-path against the discs in $\overline{\mathscr{D}}$ in $O(n)$ time.

Now, we define the \textit{weight function} $\weight: E_s \cup \overline{E_s} \times T \rightarrow {\rm I\!R}^+$, where $\weights{e}{\tau}$ is the length of the robot-path $\robotpath(e, \tau)$ if it is valid, and $\infty$ otherwise:
\begin{align} \label{weight-eq}
        \weights{e}{\tau} = \left\{  
          \begin{array}{ll}
                    |\robotpath(e, \tau)| & \text{if $\robotpath(e, \tau)$ is valid}\\
                    \infty & \text{if $\robotpath(e, \tau)$ is invalid}\\
                  \end{array} 
                \right.
\end{align}
We assign $\infty$ weight to the invalid robot-paths specifically to stop our algorithm from including them in the final answer.

Using the weight function and the adjacency graph $\overline{G}$, we devise Algorithm \ref{basic-algorithm}
to find a query time-minimal path among growing discs, from some start configuration
$s$ to a goal configuration $d$. 
In line 4 of this algorithm, we run Dijkstra's single source shortest-path algorithm \cite{Dijkstra} on $\overline{G}$, where the weight of each edge is defined as follows. Let $e=(u,v)$ be an edge in $\overline{G}$, where the minimum collision-free distance from $s$ to $u$, denoted by $dist(u)$, has already been calculated. Then, we define the weight of edge $e$ as $\weights{e}{{dist(u) \over \vr}}$, where ${dist(u) \over \vr}$ is the departure time at vertex $u$. In Lemma \ref{dijk-proof}, we prove that the output of this algorithm is a valid time-minimal path.

\begin{algorithm}[H]
\caption{Query Time-Minimal Path Algorithm} \label{basic-algorithm}

\begin{description}
\item {\bf Input}: Query points $s$ and $d$, adjacency graph $G$ and blocked set $\invalidinterval$.
\item {\bf Output}: A collision free time-minimal path from $s$ to $d$.
\end{description}
\begin{algorithmic}[1]
\STATE let $\overline{\mathscr{D}} = \mathscr{D} \cup \{s, d\}$, where $s$ and $d$ are two discs with zero radius and zero velocities
\STATE construct the adjacency graph $\overline{G}$ corresponding to set $\overline{\mathscr{D}}$
\STATE assign weight function $\weight$ to the edges in $\overline{G}$ 
\STATE \label{dij}let $\pi(s, d)$ be the shortest path reported by Dijkstra's algorithm  
\STATE return $\robotpath(\pi(s, d), 0)$
\end{algorithmic}
\end{algorithm}

\begin{lemma}\label{dijk-proof}
Algorithm \ref{basic-algorithm} computes a time-minimal path from $s$ to $d$.
\end{lemma}
\begin{proof}
We wish to show that if there exists a feasible solution, Algorithm \ref{basic-algorithm} returns a time-minimal path. Let $\T$ be robot-path reported by the algorithm, where $|\T|$ represents its length. By contradiction, we assume there exists a collision free time-minimal path $\T^*$, where $|\T^*|<|\T|$. By Lemma \ref{tang-spir}, we can represent $\T^*$ by a sequence of sub-paths $\Big(s_1, ..., s_h\Big) \subseteq \ee$. Note that, for each tangent or spiral path in $\ee$, there exist an associated edge in the adjacency graph. Thus, for the robot-path $\T^* = \Big(s_1, ..., s_h\Big)$, there is a path in the graph, namely $\pi = \Big(e_1, e_2 ..., e_h\Big)$, where each edge $e_i \in \pi$ is associated with $s_i \in \T^*$. Let $e_i=(v_{i-1}, v_i)$, where $v_0 = s$ and $v_h = d$.
Since $\T^*$ is a valid robot-path, by Equations \ref{robot-path-size} and \ref{weight-eq}, the weight of path $\pi$ is
 $W_h = \sum_{i=1}^{h} {|\robotpath(e_i, \tau_i))|}$, where $\tau_i = \tau_{i-1} + {|\robotpath(e_i, \tau_i)| \over \vr}$ and $\tau_1 = 0$. Now, for a given vertex $v_j$, let $dist(v_j)$ be the distance from $v_0$ (the source) to $v_j$, reported by Dijkstra's algorithm (line \ref{dij}).
According to our assumption, 
\begin{align*}
W_h < dist(v_h) = |\T|.
\end{align*} 
By the definition, we have $dist(v_{h}) < dist(v_{h - 1}) + |\robotpath(e_h, \tau_h)|$. Thus,
\begin{align*}
 W_h & = W_{h - 1} + |\robotpath(e_h, \tau_h)| \\ &< dist(v_h) \\
 & < dist(v_{h-1}) + |\robotpath(e_h, \tau_h)|. 
\end{align*}
So, we have 
\begin{align*}
W_{h - 1} < dist(v_{h-1}).
\end{align*}
Inductively we can generalize this inequality for every vertex $v_i$ where $1 \leq i \leq h$: $W_{i} < dist(v_{i})$.
 Therefore, we have $W_{1}  < dist(v_{1})$, which is a contradiction because $W_{1}  = dist(v_{1}) = |\robotpath(e_1, \tau_1)|$.
$\Box$
\end{proof}
\begin{lemma}\label{basic-proof}
Algorithm \ref{basic-algorithm} runs in $O(n^2 \log (kn))$ time.
\end{lemma}
\begin{proof}
% As suggested by Algorithm \ref{pre-algorithm}, the preprocessing phase of algorithm runs in $O(n^2 \log n + k \log k)$ time.
Because Dijkstra's algorithm processes each edge at most once, there are $O(n^2)$ total calls to the weight function in line 4 of the algorithm. Refer to Equation \ref{weight-eq}, computing a weight value $\weights{e}{\tau}$ requires determining if $\robotpath(e, \tau)$ is a valid path. We consider two cases: (1) if $e \in E_s$,
by Lemmas \ref{valid-time-b} and \ref{valid-time-d}, the validity of a robot-path can be determined in $O(\log k + deg(\mathcal{H}_i)+ deg(\mathcal{H}_j))$ time. Thus, the total time complexity for validating edges in $E_s$ (in line 4) is $O(n^2 \log k + 2n \sum_{i = 0}^{n}deg(\mathcal{H}_i))$, which is equivalent to $O(n^2 \log k + n^2)$. (2) if $e \in \overline{E_s}$ then the validity of $\robotpath(e, \tau)$ is determined in $O(n)$ by comparing the robot-path with every disc in $\overline{\mathscr{D}}$. Since there are $O(n)$ edges in $\overline{E_s}$, the total time complexity is equal to $O(n^2)$.

Since Dijkstra's algorithm runs in $O(n^2 \log n)$ time, the time complexity of Algorithm \ref{basic-algorithm}  is $O(n^2 (\log k + \log n)) = O(n^2 \log (kn))$.
$\Box$
\end{proof} 

\begin{theorem}\label{algorithm-theorem}
The query time-minimal path problem among growing discs can be solved in $O(n^2 \log (kn))$ time, after $O(n^2 \log n + k \log k)$ preprocessing time, where $k = O(n^3\dd)$. 
\end{theorem}
\begin{proof}
This is a direct consequence of Lemma \ref{basic-proof} and Property \ref{maxip}.
$\Box$
\end{proof}

\section{Concluding Remarks}\label{section-conclusions}
In this section, we studied the time-minimal path problem among growing discs. We presented an algorithm with an $O(n^2 \log (kn))$ query time after $O(n^2 \log n + k \log k)$ preprocessing time, where $n$ is the number of discs and $k$ is the arrangement size. 
 The best known algorithm for the restricted shortest path problem with static discs obstacles (when $\v{i}(t) = 0$) \cite{chang}, runs in $O(n^2 \log (kn))$ time which matches the time complexity of our algorithm when $\beta$ (the degree of the velocity functions) is constant ($k = O(n)$). 
 % Thus, when $k \in O(n^3 \dd)$, our algorithm is optimal, unless the shortest path algorithm for disc obstacles is improved.
% For future research, we propose and solve the QTMP problem for the case when the obstacles appear and disappear in the environment. Such an obstacle is called a "transient obstacle". 
% We also present an algorithm for the time-minimal $L_1$ path among transient disc obstacles. 
We leave as an open problem to find an improved algorithm which solves QTMP in $O(n^2 \log n)$. 
It is also interesting to modify our algorithm to solve the time-minimal path problem among shrinking discs. This can not be achieved only by negating the weight function. Note that since the robot may wait for some disc(s) to shrink, Observation \ref{max-velocity} is not valid anymore. Another open problem is to determine a time-minimal path among a set of discs $\mathscr{D}=\{C_1, ..., C_n\}$ where each disc $C_i$ either grows or shrinks with velocity $\v{i}(t)$.

\section*{Acknowledgments}
We would like to thank Anil Maheshwari for co-supervising this research as part of a PhD thesis.

\bibliographystyle{unsrt}
\bibliographystyle{plain}
\addcontentsline{toc}{section}{References}
\bibliography{thesis}

\begin{thebibliography}{10}

\bibitem{shortest_path_in_growing_disc_referred}
Jur van~den Berg and Mark Overmars.
\newblock {\em Planning the Shortest Safe Path Amidst Unpredictably Moving
  Obstacles}, pages 103--118.
\newblock Springer Berlin Heidelberg, Berlin, Heidelberg, 2008.

\bibitem{growing_discs_overmars}
J.~van~den Berg.
\newblock Path planning in dynamic environments.
\newblock {\em Ph.D. Thesis, Utrecht Uni- versity, Utrecht, The Netherlands},
  2007.

\bibitem{yi-thesis}
Jiehua Yi.
\newblock {\em A Ubiquitous Gis: Framework, Services and Algorithms
  Development}.
\newblock PhD thesis, Ottawa, Ont., Canada, Canada, 2009.
\newblock AAINR52081.

\bibitem{Mitchell1}
J.S.B Mitchell.
\newblock Geometric shortest paths and network optimization.
\newblock {\em Handbook of Computational Geometry, eds.,
  {J\"{o}rg}-{R\"{u}diger} Sack and Jorge Urrutia, Elsevier}, pages 633--701,
  2000.

\bibitem{Shortest-path-and-networks}
J.~S.~B. Mitchell.
\newblock Shortest path and networks.
\newblock {\em Handbook of Discrete and computational geometry}, page
  755–778, 1997.

\bibitem{planning-algorithms}
S.~M. LaValle.
\newblock Planning algorithms.
\newblock {\em Chapter 2, Cambridge University Press}, 2006.

\bibitem{reif}
J.~Reif and M.~Sharir.
\newblock Motion planning in the presence of moving obstacles.
\newblock In {\em 26th Annual Symposium on Foundations of Computer Science
  (sfcs 1985)}, pages 144--154, Oct 1985.

\bibitem{erdmann}
M.~Erdmann and T.~Lozano-Perez.
\newblock On multiple moving objects.
\newblock In {\em Proceedings. 1986 IEEE International Conference on Robotics
  and Automation}, volume~3, pages 1419--1424, Apr 1986.

\bibitem{vandenberg}
J.~van~den Berg, D.~Ferguson, and J.~Kuffner.
\newblock Anytime path planning and replanning in dynamic environments.
\newblock In {\em Proceedings 2006 IEEE International Conference on Robotics
  and Automation, 2006. ICRA 2006.}, pages 2366--2371, May 2006.

\bibitem{fiorini}
Paolo Fiorini and Zvi Shiller.
\newblock Motion planning in dynamic environments using velocity obstacles.
\newblock {\em The International Journal of Robotics Research}, 17(7):760--772,
  1998.

\bibitem{high-speed-autonomous}
D.~Vasquez, F.~Large, T.~Fraichard, and C.~Laugier.
\newblock High-speed autonomous navigation with motion prediction for unknown
  moving obstacles.
\newblock In {\em 2004 IEEE/RSJ International Conference on Intelligent Robots
  and Systems (IROS) (IEEE Cat. No.04CH37566)}, volume~1, pages 82--87 vol.1,
  Sept 2004.

\bibitem{path-finding}
J.~van~den Berg.
\newblock Path planning in dynamic environments.
\newblock {\em Ph.D. Thesis, Utrecht University, Utrecht, The Netherlands},
  2007.

\bibitem{chang}
Ee-Chien Chang, Sung~Woo Choi, DoYong Kwon, Hyungju Park, and Chee~K. Yap.
\newblock Shortest path amidst disc obstacles is computable.
\newblock In {\em Proceedings of the Twenty-first Annual Symposium on
  Computational Geometry}, SCG '05, pages 116--125, New York, NY, USA, 2005.
  ACM.

\bibitem{Spatial-tessellations}
A.~Okabe, B.~Boots, K.~Sugihara, and S.~N. Chiu.
\newblock Spatial tessellations: concepts and applications of {V}oronoi
  diagrams, 2nd edition.
\newblock 2000.

\bibitem{Bentley-Ottmann}
J.L.Bentley and T.Ottmann.
\newblock Algorithms for reporting and counting geometric intersections.
\newblock {\em IEEE Transactions on Computers}, pages 643--647, 1979.

\bibitem{point-location}
M.~Shimrat.
\newblock Algorithm 112: Position of point relative to polygon.
\newblock {\em Commun. ACM}, 5(8):434--, August 1962.

\bibitem{Dijkstra}
E.~W. Dijkstra.
\newblock A note on two problems in connexion with graphs.
\newblock {\em Numerische Mathematik}, 1(1):269--271, Dec 1959.

\end{thebibliography}

 % \input{9-Appendix}

% ============================================================================
\end{document}